\documentclass[a4paper,oneside,11pt,reqno]{amsart}
\usepackage{amssymb}
\usepackage{amsthm}
\usepackage{amsmath}
\usepackage{url}
\usepackage{hyperref}
\usepackage{cite}
\usepackage{graphicx}
\usepackage[margin=1in]{geometry}
\usepackage[foot]{amsaddr}
\usepackage{comment}
\usepackage{xcolor}
\usepackage{enumitem}

\usepackage{mathtools}

\newtheorem{theorem}{Theorem}
\newtheorem{corollary}[theorem]{Corollary}
\newtheorem{lemma}[theorem]{Lemma}
\newtheorem{proposition}[theorem]{Proposition}
\theoremstyle{remark}

\newcommand{\ea}[1]{\begin{align}#1\end{align}}

\newcommand{\la}{\langle}
\newcommand{\ra}{\rangle}
\newcommand{\mcl}{\mathcal}

\newcommand{\mbb}{\mathbb}
\newcommand{\mrm}{\mathrm}
\newcommand{\bmt}{\begin{pmatrix}}
\newcommand{\emt}{\end{pmatrix}}

\DeclareMathOperator{\tr}{Tr}
\DeclareMathOperator{\ebr}{ebr}
\DeclareMathOperator{\len}{len}
\DeclareMathOperator{\ch}{\mcl{C}}
\DeclareMathOperator{\rank}{Rank}

\numberwithin{equation}{section}

\begin{document}

\title{On the continuous Zauner conjecture}

\author{Danylo Yakymenko}
\email{danylo.yakymenko@gmail.com}  
\address{Institute of mathematics of the National Academy of Sciences of Ukraine, Kyiv, Ukraine}

\begin{abstract}
In a recent paper by S.Pandey, V.Paulsen, J.Prakash, and M.Rahaman, the authors studied the entanglement breaking quantum channels $\Phi_t:\mbb{C}^{d\times d} \to \mbb{C}^{d \times d}$ for $t \in [-\frac{1}{d^2-1}, \frac{1}{d+1}]$ defined by $\Phi_t(X) = tX+ (1-t)\tr(X) \frac{1}{d}I$. They proved that Zauner's conjecture is equivalent to the statement that entanglement breaking rank of $\Phi_{\frac{1}{d+1}}$ is $d^2$. The authors made the extended conjecture that $\ebr(\Phi_t)=d^2$ for every $t \in [0, \frac{1}{d+1}]$ and proved it in dimensions 2 and 3. 

In this paper we prove that for any $t \in [-\frac{1}{d^2-1}, \frac{1}{d+1}] \setminus\{0\}$ the equality $\ebr(\Phi_t)=d^2$ is equivalent to the existence of a pair of informationally complete unit norm tight frames $\{|x_i\ra\}_{i=1}^{d^2}, \{|y_i\ra\}_{i=1}^{d^2}$  in $\mbb{C}^d $ which are mutually unbiased in a certain sense. That is, for any $i\neq j$ it holds that $|\la x_i|y_j\ra|^2 = \frac{1-t}{d}$ and $|\la x_i|y_i\ra|^2 = \frac{t(d^2-1)+1}{d}$ (also it follows that $|\la x_i|x_j\ra\la y_i|y_j\ra|=|t|$). 

Though, our numerical searches for solutions were not successful in dimensions 4 and 5 for values of $t$ other than $0$ or $\frac{1}{d+1}$. 
 
\end{abstract}

\maketitle
\allowdisplaybreaks

\section{Introduction}

Recall that a \textit{tight frame} \cite{Waldron} is a set of vectors $\{|x_i\ra\}_{i=1}^n \in \mbb{C}^{d}$ such that
\ea{\label{eq:int-1}
	\sum_{i=1}^{n} |x_i \ra\la x_i | = cI_d, 
}
for some real $c \ge 0$. 
A frame is called \textit{unit norm} if $|x_i\ra$ are unit vectors. 
In this paper we consider only unit norm frames. By calculating the trace of Eq. \ref{eq:int-1} one can see that $c = \frac{n}{d}$ for a tight frame. 
A frame is \textit{equiangular} if $|\la x_i | x_j \ra|$ is constant for any $i \neq j$. 
An equiangular tight frame in $\mbb{C}^{d}$ can't have more than $d^2$ elements. 
In the extreme case it's called a \textit{maximal ETF}, 
and one can deduce that $|\la x_i | x_j \ra|^2 = \frac{1}{d+1}$ for any $i \neq j$ in this case. 
We'll also call such a frame by acronym \textit{SIC}, 
since the set $\{ \frac{1}{d} |x_i\ra\la x_i | \}_{i=1}^{d^2}$ is known as 
a symmetric, informationally complete, positive operator-valued measure (SIC-POVM) \cite{Renes}.
The informational completeness means that the linear span of projectors $|x_i \ra\la x_i |$ is the set of all complex $d \times d$ matrices. 

Another distinctive property of a SIC is that it forms a \textit{minimal projective 2-design}, that is
\ea{\label{eq:int-2}
	\frac{1}{d^2}\sum_{i=1}^{d^2} \Big(|x_i \ra\la x_i |\Big)^{\otimes 2}  
	= \int_{\mbb{C}P^{d-1}} \Big(|\phi \ra\la \phi |\Big)^{\otimes 2} \mathrm{d}\mu(\phi) 
	= \frac{2}{d(d+1)}\Pi_{sym}, 
}
where $\mbb{C}P^{d-1}$ is the corresponding projective space of $\mbb{C}^d$ (i.e. lines passing through the origin), 
$\mu$ denotes the unique unitarily-invariant probability measure on $\mbb{C}P^{d-1}$ 
induced by the Haar measure on the group of unitary matrices $\mcl{U}(d)$, 
and $\Pi_{sym}$ is the projector onto the symmetric subspace of $(\mbb{C}^d)^{\otimes 2}$
(see \cite{Scott} for a review). 

It's remarkable that such an optimal and symmetric structure presumably exists in every dimension $d$, but the general proof is still elusive. The existence of SICs for every $d$ is known as Zauner's weak conjecture, due to its first major study in \cite{Zauner}. There is an extensive literature devoted to the research of this structure and related questions, see, for example, \cite{STFF, SG, Aligned, Yard, Bengtsson, Kopp, Zhu, Scott2, Quaternions}. For a survey we refer to \cite{Fuchs}. 

\newpage

In a recent paper by S.Pandey, V.Paulsen, J.Prakash, and M.Rahaman \cite{PPPR}, the authors studied the entanglement breaking quantum channels $\Phi_t:\mbb{C}^{d\times d} \to \mbb{C}^{d \times d}$ for $t \in [-\frac{1}{d^2-1}, \frac{1}{d+1}]$ defined by 
\ea{
\Phi_t(X) = tX+ (1-t)\tr(X) \frac{1}{d}I_d
}
(we review entanglement breaking maps in Section \ref{sec:ebr}). 
They proved that Zauner's weak conjecture is equivalent to the statement that entanglement breaking rank ($\ebr$ for short) of $\Phi_{\frac{1}{d+1}}$ is $d^2$. On the other hand, $\ebr(\Phi_0)=d^2$ is true in any dimension. The authors made the extended conjecture, which we refer to as the \textit{PPPR conjecture}, that $\ebr(\Phi_t)=d^2$ for every $t \in [0, \frac{1}{d+1}]$ and proved it in dimensions 2 and 3. It seems to be true for $t \in [-\frac{1}{d^2-1},0)$ in these dimensions as well, though we are more interested in the case of general dimension $d$. 

In this paper our main results are contained in Section \ref{sec:muf}, where we prove that for any $t \in [-\frac{1}{d^2-1}, \frac{1}{d+1}] \setminus\{0\}$ the equality $\ebr(\Phi_t)=d^2$ is equivalent to the existence of a pair of informationally complete unit norm tight frames $\{|x_i\ra\}_{i=1}^{d^2}, \{|y_i\ra\}_{i=1}^{d^2}$  in $\mbb{C}^d $ which are mutually unbiased in a certain sense. 
That is, for any $i\neq j$ it holds that 
$$
	|\la x_i|y_j\ra|^2 = \frac{1-t}{d}, 
$$
$$	
	|\la x_i|y_i\ra|^2 = \frac{t(d^2-1)+1}{d},
$$
\ea{
	|\la x_i|x_j\ra\la y_i|y_j\ra|=|t|.
} 
We'll call such a pair of frames as \textit{MUF}. 
One can readily see that for $t=\frac{1}{d+1}$ the two frames should coincide (up to phase factors) and satisfy SIC relations. 
Moreover, the results of \cite{PPPR} imply that there exists a family of MUFs which is continuous over $t$ on $[0, \frac{1}{d+1}]$ in dimensions 2 and 3 (with a nuance at $t=0$). Therefore, the PPPR conjecture may be seen as a continuous extension of Zauner's conjecture. 

In Section \ref{sec:t0} we investigate the case where $t=0$. 

Even though the existence of such pairs of frames may look plausible, our numerical searches were not successful in dimensions 4 and 5 for values of $t$ other than $0$ or $\frac{1}{d+1}$. 

\section{Entanglement breaking maps}
\label{sec:ebr}

A linear map $\Phi: \mbb{C}^{d\times d} \to \mbb{C}^{d \times d}$ is called a \textit{quantum channel} if it's trace preserving and completely positive \cite{Watrous}, 
that is, $\tr(\Phi(X))=\tr(X)$ for any matrix $X$, and $\mcl{I}_n \otimes \Phi$ is a positive map for any $n$, where $\mcl{I}_n$ is the identity map on $\mbb{C}^{n\times n}$. 

Consider the rank 1 projector in the space of matrices $\mbb{C}^{d \times d} \otimes \mbb{C}^{d \times d}$:
\ea{
	|\beta\ra\la\beta| 
	= \frac{1}{d}\Big(\sum_{i=0}^{d-1} |i\ra|i\ra\Big)\Big(\sum_{j=0}^{d-1} \la j|\la j|\Big)
	= \frac{1}{d}\sum_{i=0}^{d-1} \sum_{j=0}^{d-1} E_{ij} \otimes E_{ij},
}
where $\{|i\ra\}_{i=0}^{d-1}$ is the standard orthonormal basis and $E_{ij} = |i\ra\la j|$ are matrix units.

The image of it under the $\mcl{I}_d \otimes \Phi$ map is called the \textit{Choi matrix} of $\Phi$:
\ea{
	\mcl{C}(\Phi) = (\mcl{I}_d \otimes \Phi)(|\beta\ra\la\beta|) = \frac{1}{d}\sum_{ij} E_{ij} \otimes \Phi(E_{ij}).
} 

By Choi's theorem \cite{Choi}, $\Phi$ is completely positive if and only if $\mcl{C}(\Phi)$ is positive semidefinite. 
The factor $\frac{1}{d}$ ensures that $\tr(\mcl{C}(\Phi))=1$ for a trace preserving $\Phi$, although it's common to define Choi matrices without it.
The Choi map $\mcl{C}$ is of great importance in quantum information theory because it's linear and invertible, 
i.e. $\mcl{C}(\Phi)$ completely describes $\Phi$. 

For example, for the map $\Phi_0(X) = \frac{1}{d}\tr(X)I_d$ we have $\mcl{C}_0 = \mcl{C}(\Phi_0) = \frac{1}{d^2}I_{d^2}$, 
and for the map $\Phi_1(X) = X$ we have $\mcl{C}_1 = \mcl{C}(\Phi_1) = |\beta\ra\la\beta|$. 
One can see that the linear combination $t\mcl{C}_1 + (1-t)\mcl{C}_0$ is positive semidefinite for $t \in [-\frac{1}{d^2-1},1]$, 
thus $\Phi_t = t\Phi_1 + (1-t)\Phi_0$ are legit quantum channels on this interval. 
They're known as \textit{quantum depolarizing channels}. 

A particular naturally occurring class of quantum channels is known as \textit{entanglement breaking channels} \cite{ebc}. 
Those are such $\Phi$ that for any state $\rho \in \mbb{C}^{n \times n} \otimes \mbb{C}^{d \times d}$ the image $(\mcl{I}_n \otimes \Phi)(\rho)$ is always disentangled (separable). 
Recall that a state $\rho$ on a bipartite system is called \textit{separable} if it can be written as a convex combination of product states, that is 
$$
\rho = \sum_{i=1}^m \lambda_i \rho_i^{(1)} \otimes \rho_i^{(2)},
$$
where $\lambda_i \ge 0$, $\sum_i \lambda_i = 1$, and $\rho_i^{(1)}, \rho_i^{(2)}$ are states on the corresponding subsystems. 
On can see that if such a decomposition exists, then there is also a decomposition where each $\rho_i^{(1)}, \rho_i^{(2)}$ are pure states, i.e. rank 1 projectors. A separable state can have many different pure decompositions. 
The minimum number $m$ of the summands that a pure decomposition can have is called the \textit{length of separability} of $\rho$, 
which we denote by $\len(\rho)$. 
For example, it's easy to see that $\len(\mcl{C}_0) = d^2$ where for the decomposition one can take 
$\frac{1}{d^2}I_{d^2} = \frac{1}{d^2} \sum_{ij} |i\ra\la i| \otimes |j\ra\la j|$. 

It turns out that $\Phi$ is entanglement breaking if it "breaks" just the single element $|\beta\ra\la\beta|$ (it's maximally entangled, though). In other words, $\Phi$ is entanglement breaking if and only if its Choi matrix $\mcl{C}(\Phi)$ is separable. 
In such a case, what was defined in \cite{PPPR} as \textit{entanglement breaking rank} of $\Phi$ equals to the length of separability of $\mcl{C}(\Phi)$, 
that is, 
$$
\ebr(\Phi) = \len(\mcl{C}(\Phi)),
$$ 
so that these two notions can be used interchangeably. 
The original definition uses the notion of Kraus decompositions, 
but in this paper we only concern Choi matrices.

One can recognise that the matrices $\mcl{C}_t = t\mcl{C}_1 + (1-t)\mcl{C}_0$ are also known as \textit{isotropic states} (see \cite{Watrous}, Example 6.10, Example 7.25). 
It's known that $\mcl{C}_t$ are separable for $t \in [-\frac{1}{d^2-1}, \frac{1}{d+1}]$ and entangled otherwise. 
In the separable case their partial transpose are separable Werner states (clearly, the partial transpose doesn't affect separability and the length of separability remains the same). 

The calculation of partial transpose gives
$$
(\mcl{I} \otimes T)(\mcl{C}_0) = \frac{1}{d^2}I_{d^2},
$$
\ea{
	(\mcl{I} \otimes T)(\mcl{C}_1) = \frac{1}{d} \sum_{ij} E_{ij} \otimes E_{ij}^T = \frac{1}{d} U_{SW}, 
}
where $U_{SW}$ is the swap operator on $\mbb{C}^{d} \otimes \mbb{C}^{d}$ 
and $T(X)=X^T$ is the transpose operator on $\mbb{C}^{d \times d}$. 
Note that $\Pi_{sym} = \frac{1}{2}(I_{d^2} + U_{SW})$ is the projector onto the symmetric subspace of $\mbb{C}^{d} \otimes \mbb{C}^{d}$, 
while $\Pi_{asym} = \frac{1}{2}(I_{d^2} - U_{SW})$ is the projector onto the asymmetric subspace. 
So that, in particular, we have
$$
	(\mcl{I} \otimes T)(\mcl{C}_{\frac{1}{d+1}}) 
	= \frac{1}{d(d+1)} U_{SW} + \frac{1}{d(d+1)} I_{d^2} 
	= \frac{2}{d(d+1)} \Pi_{sym}, 
$$
\ea{
	(\mcl{I} \otimes T)(\mcl{C}_{\frac{-1}{d^2-1}}) 
	= \frac{-1}{d(d^2-1)} U_{SW} + \frac{1}{d^2-1} I_{d^2} 
	= \frac{1}{d(d+1)}\Pi_{sym} + \frac{1}{d(d-1)} \Pi_{asym}. 
}

Less is known about their separability length. Clearly, it can't be smaller than the rank of matrix $\mcl{C}_t$. 
It's not hard to see that either $\rank(\mcl{C}_t) = d^2$ or $\rank((\mcl{I} \otimes T)(\mcl{C}_t)) = d^2$
for any $t$, thus $\len(\mcl{C}_t) \ge d^2$. 
On the other hand, from the Caratheodory theorem it can be deduced that $\len(\mcl{C}_t) \le d^4$.
For $t=\frac{1}{d+1}$ it gives a bit better boundary, $\len(\mcl{C}_{\frac{1}{d+1}}) \le d^2(d+1)^2/4$. 
See \cite{Mixon}, Corollary 5, for a collection of other specific boundaries. 
For example, it's known that a maximal set of $d+1$ mutually unbiased bases (MUB for short, \cite{mub}) in $\mbb{C}^d$ is a projective 2-design \cite{Rotteler}. This means that the sum of the tensor squares of MUB elements gives a pure decomposition of $\Pi_{sym}$ (normalized approprietly). It follows that $\len(\mcl{C}_{\frac{1}{d+1}}) \le d(d+1)$ if $d$ is a prime power, since a maximal set of MUBs exists in this case. 

In turn, SICs are minimal projective 2-designs.   
It follows from Eq. \ref{eq:int-2} that $\ebr(\Phi_{\frac{1}{d+1}}) = \len\big(\mcl{C}_{\frac{1}{d+1}}\big) = \len\big((\mcl{I} \otimes T)(\mcl{C}_{\frac{1}{d+1}})\big) = d^2$ whenever SIC exists in dimension $d$. 
In \cite{PPPR}, Theorem III.2, the authors proved the backward implication. 
That is, if $\ebr(\Phi_{\frac{1}{d+1}}) = d^2$ then a SIC exists. 
Therefore, Zauner's weak conjecture is equivalent to the statement that $\ebr(\Phi_{\frac{1}{d+1}}) = d^2$. 
In Section \ref{sec:muf} we show that a similar equivalence can be proved for $\Phi_t$ in general. 

\section{Mutually unbiased frames}
\label{sec:muf}
In this section we present our main results about the relation between entanglement breaking rank of considered maps and the existence of a pair of mutually unbiased frames. 

\begin{theorem}
\label{th:muf-1}
Let $t \in [-\frac{1}{d^2-1}, \frac{1}{d+1}] \setminus \{0\}$ and suppose that $\len\big((\mcl{I} \otimes T)(\mcl{C}_t)\big) = d^2$. That is, 
\ea{\label{eq:muf-1}
\sum_{i=1}^{d^2} w_i |x_i\ra\la x_i| \otimes |y_i\ra\la y_i| 
= \frac{2t}{d} \Pi_{sym} + \frac{1-t(d+1)}{d^2}I_{d^2} 
= \frac{t}{d}U_{SW} +\frac{1-t}{d^2}I_{d^2} 
}
for some real weights $w_i \ge 0$ that sum to $1$ and some unit vectors $\{|x_i\ra\}_{i=1}^{d^2}, \{|y_i\ra\}_{i=1}^{d^2} \subset \mbb{C}^d$. 

Then all weights are equal, i.e. $w_i = \frac{1}{d^2}$, the frames $\{|x_i\ra\}$, $\{|y_i\ra\}$ are tight and informationally complete, 
and for any $i \neq j$ we have 
\ea{\label{eq:muf-2}
|\la x_i|y_j\ra|^2 = \frac{1-t}{d}, \quad \la x_i|y_i\ra = \sqrt{\frac{t(d^2-1)+1}{d}}, 
}
\ea{\label{eq:muf-3}
\la x_j|x_i\ra\la y_i|y_j\ra=t,
}
under a suitable choice of phases of $\{|x_i\ra\}$, $\{|y_i\ra\}$. 

\end{theorem}

\begin{proof}
Note that since Eq. \ref{eq:muf-1} doesn't depend on the phases of vectors $\{|x_i\ra\}, \{|y_i\ra\}$ we are free to choose any. 

We proceed similarly to the proof of Theorem 4 in \cite{Scott}. Let's denote $\pi_i = |x_i\ra\la x_i|$, $\rho_i = |y_i\ra\la y_i|$. 

By taking the partial traces of Eq. \ref{eq:muf-1} we get
\ea{\label{eq:muf-4}
\sum_{i=1}^{d^2} w_i \pi_i = \sum_{i=1}^{d^2} w_i \rho_i = \frac{1}{d}I_d,
}

since $\tr_1(U_{SW}) = \tr_1( \sum_{ij} |i\ra \la j| \otimes |j\ra\la i|) = \sum_{ij} \la j|i\ra |j\ra \la i| = I_d$, similarly $\tr_2(U_{SW}) = I_d$. 
Thus  $\{|x_i\ra\}$ and $\{|y_i\ra\}$ are weighted tight frames. 

For any matrix $A \in \mbb{C}^{d\times d}$ we can multiply Eq. \ref{eq:muf-1} by $A \otimes I$ and take the partial trace $\tr_1$. 
This gives us
\ea{\label{eq:muf-5}
\sum_{i=1}^{d^2} w_i \tr(\pi_iA)\rho_i = \frac{t}{d}A + \frac{1-t}{d^2}\tr(A)I_d, 
}
since $\tr_1(U_{SW} \cdot A\otimes I) = \tr_1( \sum_{ij} |i\ra \la j|A \otimes |j\ra\la i|) = \sum_{ij} \tr\big(A\cdot (|j\ra\la i|)^\dag \big)|j\ra\la i| = A$.

For $t \neq 0$ we can see that $A$ is a linear combination of $\rho_i$. Thus the set $\{\rho_i\}$ is informationally complete. The same can be deduced for the set $\{\pi_i\}$. By using the tightness of $\{\rho_i\}$ (Eq. \ref{eq:muf-4}) we can deduce the reconstruction formula 
\ea{\label{eq:muf-6}
A = \frac{d}{t}\bigg( \sum_{i=1}^{d^2}w_i (\tr(\pi_i A) - \frac{1-t}{d}\tr(A)) \rho_i\bigg).
}

Now let's substitute $A$ for $\rho_j$ for all $j=1,\dots,d^2$. We get
\ea{\label{eq:muf-7}
\rho_j = \frac{d}{t}\bigg( \sum_{i=1}^{d^2}w_i (\tr(\pi_i \rho_j) - \frac{1-t}{d}) \rho_i\bigg).
}

Since the set $\{\rho_i\}_{i=1}^{d^2}$ is informationally complete it's also linearly independent in the space of all matrices. 
Thus from Eq.\ref{eq:muf-7} we deduce that for any $i \neq j$:
\ea{\label{eq:muf-8}
\tr(\pi_i \rho_j) = |\la x_i|y_j\ra|^2 = \frac{1-t}{d},
}
\ea{\label{eq:muf-9}
\tr(\pi_j \rho_j) = |\la x_j|y_j\ra|^2 = \frac{t}{dw_j} + \frac{1-t}{d}.
}

At this point for $t=-\frac{1}{d^2-1}$ we can prove that all weights $w_i$ must be equal. Indeed, $|\la x_j|y_j\ra|^2 \geq 0$, thus 
$-\frac{1}{(d^2-1)dw_j} + \frac{d^2}{(d^2-1)d} \geq 0$, which gives $w_j \geq \frac{1}{d^2}$ for all $j$. Since $\sum_j w_j =1$ we deduce that $w_j = \frac{1}{d^2}$ for all $j$. 

Let's consider our starting equation Eq. \ref{eq:muf-1} and multiply it from the right by $U_{SW}$. We get
\ea{\label{eq:muf-10}
\sum_{i=1}^{d^2} w_i |x_i\ra\la y_i| \otimes |y_i\ra\la x_i| 
= \frac{t}{d}I_{d^2} +\frac{1-t}{d^2}U_{SW}. 
} 

By the way, if we multiply it again by $U_{SW}$ from the left we obtain Eq. \ref{eq:muf-1} but with swapped symbols $x,y$. This means that if a pair of frames is a solution to Eq. \ref{eq:muf-1} then the swapped pair is also a solution.

Now we use our previous trick. For any matrix $B \in \mbb{C}^{d\times d}$ multiply Eq.\ref{eq:muf-10} by $B \otimes I$ and take the partial trace $\tr_1$. This gives us 
\ea{\label{eq:muf-11}
\sum_{i=1}^{d^2} w_i \tr(|x_i\ra\la y_i| B) |y_i\ra\la x_i| = \frac{t}{d}\tr(B)I_d + \frac{1-t}{d^2}B. 
}

From substitution $B=I_d$ we get

\ea{\label{eq:muf-12}
\sum_{i=1}^{d^2} w_i \la y_i | x_i\ra \cdot  |y_i\ra\la x_i| = \frac{t(d^2-1) + 1}{d^2}I_d. 
}

It's not hard to see that the map $\frac{t}{d}\tr(B)I_d + \frac{1-t}{d^2}B$ is invertible unless $t=-\frac{1}{d^2-1}$, in which case the image of this map is the subspace of traceless matrices. 

Thus, for $t \neq -\frac{1}{d^2-1}$ the set $\{|y_i\ra\la x_i|\}_{i=1}^{d^2}$ is linearly independent in the space of all matrices, otherwise its linear span has dimension $d^2-1$.

Using the substitutions $B=|y_j\ra\la x_j|$ in Eq. \ref{eq:muf-11} for all $j$ we obtain 
\ea{\label{eq:muf-13}
\sum_{i=1}^{d^2} w_i \la x_j|x_i\ra\la y_i|y_j\ra \cdot |y_i\ra\la x_i| = \frac{t}{d}\la x_j|y_j\ra I_d + \frac{1-t}{d^2}|y_j\ra\la x_j|. 
}

For $t \neq -\frac{1}{d^2-1}$ we can use the expression for $I_d$ from Eq. \ref{eq:muf-12} to deduce
\ea{\label{eq:muf-14}
\sum_{i=1}^{d^2} w_i \big( \la x_j|x_i\ra\la y_i|y_j\ra - \frac{td}{t(d^2-1)+1}\la x_j|y_j\ra\la y_i | x_i\ra  \big) \cdot |y_i\ra\la x_i| 
=  \frac{1-t}{d^2}|y_j\ra\la x_j|.
}
From linear independence of $\{|y_i\ra\la x_i|\}$ we deduce that for any $i \neq j$:
\ea{\label{eq:muf-15}
\la x_j|x_i\ra\la y_i|y_j\ra = \frac{td}{t(d^2-1)+1}\la x_j|y_j\ra\la y_i | x_i\ra,
}
\ea{\label{eq:muf-16}
w_j (1 - \frac{td}{t(d^2-1)+1} |\la x_j|y_j\ra|^2) =  \frac{1-t}{d^2}.
}

Recall that $|\la x_j|y_j\ra|^2 = \frac{t}{dw_j} + \frac{1-t}{d}$ from Eq. \ref{eq:muf-9}. 

Let's denote 
\ea{\label{eq:muf-17}
c_j = \frac{td}{t(d^2-1)+1} |\la x_j|y_j\ra|^2.
} 

If $c_j > t$ then from Eq.\ref{eq:muf-16} we have $w_j = \frac{1-t}{d^2} \cdot \frac{1}{1-c_j} > \frac{1-t}{d^2} \cdot \frac{1}{1-t} = \frac{1}{d^2}$. 
On the other hand, from Eq. \ref{eq:muf-9} we have $\frac{1}{dw_j} =  \frac{1}{t}|\la x_j|y_j\ra|^2 - \frac{1-t}{td} = c_j \frac{t(d^2-1)+1}{t^2d} - \frac{1-t}{td} > \frac{t(d^2-1)+1}{td} - \frac{1-t}{td} = d$, 
hence $w_j < \frac{1}{d^2}$ – a contradiction. A similar contradiction we get if $c_j < t$. 
Thus, the only option left is $c_j = t$ and $w_j = \frac{1}{d^2}$ for any $j$. 

Let's choose the phases of $\{|x_i\ra\}, \{|y_i\ra\}$ such that $\la x_i|y_i\ra$ is real and non-negative for every $i$. 
From $c_j = t$ and Eq. \ref{eq:muf-15} we then deduce that for any $i \neq j$:
\ea{\label{eq:muf-18}
\la x_j|y_j\ra = \sqrt{\frac{t(d^2-1)+1}{d}},
} 
\ea{\label{eq:muf-19}
\la x_j|x_i\ra\la y_i|y_j\ra = t.
}

Finally, let's consider the remaining case $t = -\frac{1}{d^2-1}$. 
Recall that in this case we've already showed that $w_i$ must be equal $\frac{1}{d^2}$. 
Moreover, from Eq. \ref{eq:muf-9} we get $\la x_j|y_j\ra=0$ for any $j$. 

Now, for any $i \neq j$ denote $b_{ij} = \la x_j|x_i\ra\la y_i|y_j\ra$ and $b_{jj} = t$. 
Then Eq. \ref{eq:muf-13} can be rewritten as 
\ea{\label{eq:muf-20}
\sum_{i=1}^{d^2} b_{ij} |y_i\ra\la x_i| = 0. 
}
Let $M$ be a matrix with the entries $b_{ij}$. Note that $M$ is self-adjoint. 
Since the set $\{ |y_i\ra\la x_i| \}$ spans the subspace of dimension $d^2-1$, 
the rank of matrix $M$ must not exceed 1. In particular, every $2 \times 2$ minor must be equal 0.
Hence $|b_{ij}|^2=t^2$ for any $i\neq j$. 

It's easy to see that we can choose the phases of $\{|x_i\ra\}, \{|y_i\ra\}$ in such a way that $b_{12} = b_{13} = \dots = b_{1d^2} = t$. 
Hence any entry of the matrix $M$ must be equal $t$, so that we obtain the same equations Eq. \ref{eq:muf-18} and Eq. \ref{eq:muf-19} for this corner case. 

\end{proof}

A particular reverse implication is also true as we show in the next theorem. 

\begin{theorem}
Let $-\frac{1}{d^2-1} \le t \le \frac{1}{d+1}$ and suppose that we have two tight informationally complete unit frames $\{|x_i\ra\}_{i=1}^{d^2}, \{|y_i\ra\}_{i=1}^{d^2} \subset \mbb{C}^d$ such that for any $i \neq j$ it holds that
\ea{\label{eq:muf-21}
|\la x_i|y_j\ra|^2 = \frac{1-t}{d}, \quad |\la x_i|y_i\ra|^2 = \frac{t(d^2-1)+1}{d}.
}
Then 
\ea{\label{eq:muf-22}
\frac{1}{d^2}\sum_{i=1}^{d^2} |x_i\ra\la x_i| \otimes |y_i\ra\la y_i| 
= \frac{t}{d}U_{SW} +\frac{1-t}{d^2}I_{d^2}. 
}
\end{theorem}

\begin{proof}
Denote $\pi_i = |x_i\ra\la x_i|$, $\rho_i = |y_i\ra\la y_i|$.  Tightness condition means that $\sum_{i=1}^{d^2} \pi_i = \sum_{i=1}^{d^2} \rho_i = dI_d$. Informational completeness implies that the set $\{ \rho_i \otimes \pi_j \}$ spans the set of all $d^2 \times d^2$ matrices. Thus, it's enough to show that both parts of Eq. \ref{eq:muf-22} has the same Hilbert-Schmidt inner product with $\rho_k \otimes \pi_l$ for every $k,l \in \{1, \dots, d^2\}$. Let's verify. 

For $k\neq l$ we have

$$
\tr\Big( \frac{1}{d^2}\sum_{i=1}^{d^2} \pi_i \otimes \rho_i \cdot \rho_k \otimes \pi_l \Big) 
= \frac{1}{d^2}\sum_{i=1}^{d^2} \tr(\pi_i\rho_k)\tr(\rho_i\pi_l) =
$$
$$
= \frac{1}{d^2}\bigg( (d^2-2)\frac{1-t}{d} \cdot \frac{1-t}{d} + 2 \frac{1-t}{d} \cdot \frac{t(d^2-1)+1}{d}\bigg) =
$$
\ea{
= \frac{1-t}{d^4}\bigg( (d^2-2)(1-t) + 2 (t(d^2-1)+1) \bigg) 
= \frac{1-t}{d^4}(td^2 + d^2) = \frac{1-t^2}{d^2}.
}

On the other hand, 
$$
\tr\bigg( (\frac{t}{d}U_{SW} + \frac{1-t}{d^2}I_{d^2}) \cdot (\rho_k \otimes \pi_l) \bigg) =
$$
\ea{
= \frac{t}{d}\tr( |x_l\ra|y_k\ra \la y_k| \la x_l |) + \frac{1-t}{d^2} = \frac{t}{d}\cdot \frac{1-t}{d} + \frac{1-t}{d^2} = \frac{1-t^2}{d^2}. 
}

Similarly, for $k=l$ we have
$$
\tr\Big( \frac{1}{d^2}\sum_{i=1}^{d^2} \pi_i \otimes \rho_i \cdot \rho_k \otimes \pi_k \Big) 
= \frac{1}{d^2}\bigg( (d^2-1)\Big(\frac{1-t}{d}\Big)^2 + \Big(\frac{t(d^2-1)+1}{d}\Big)^2 \bigg) =
$$
\ea{
= \frac{1}{d^4}\bigg( (d^2-1)(t^2-2t+1) + t^2(d^2-1)^2 + 2t(d^2-1) + 1 \bigg) 
= \frac{t^2(d^2-1)+1}{d^2}.
}

On the other hand, 
$$
\tr\bigg( (\frac{t}{d}U_{SW} + \frac{1-t}{d^2}I_{d^2}) \cdot (\rho_k \otimes \pi_k) \bigg) =
$$
\ea{
= \frac{t}{d}\tr( |x_k\ra|y_k\ra \la y_k| \la x_k |) + \frac{1-t}{d^2} 
= \frac{t}{d}\cdot \frac{t(d^2-1)+1}{d} + \frac{1-t}{d^2}  
= \frac{t^2(d^2-1)+1}{d^2}. 
}

\end{proof}

These two theorems give us the following corollary 

\begin{corollary}
For any $t \in [-\frac{1}{d^2-1}, \frac{1}{d+1}] \setminus \{0\}$ the equality $\ebr(\Phi_t) = d^2$ is equivalent to the existence of a pair of mutually unbiased frames that appear in theorem 1.  
\end{corollary}

Thus, the PPPR conjecture is equivalent to the existence of such pairs of frames. 
For $t=\frac{1}{d+1}$ the two frames should coincide and be equal to a SIC solution.

It was proved in \cite{PPPR} that $\ebr(\Phi_t) = d^2$ for $t \in [0,\frac{1}{d+1}]$ in dimensions 2 and 3.
Their proof is constructive and one can use Lemma 8 of \cite{Mixon} to derive an explicit family of MUFs, 
which is continuous over $t$ on $[0,\frac{1}{d+1}]$, with a nuance that at $t=0$ one frame from the pair is not informationally complete. 
Also, these MUFs are Weyl-Heisenberg group covariant (we concern WH covariance in the next section). 

Even though the existence of MUFs in every dimension may look plausible, our numerical searches were not successful in dimensions 4 and 5 for values of $t$ other than $0$ or $\frac{1}{d+1}$. In dimension 4 one can try to use the construction of Groebner bases to find exact solutions or disprove the conjecture (e.g. by using the methods of \cite{Grassl}), but our attempts showed that this still requires an excessive amount of computing power. 

\section{The case where $t=0$}
\label{sec:t0}

In this section we investigate the situation where $t=0$. 

As we already mentioned, in this case $\ebr(\Phi_0) = d^2$ and for the corresponding decomposition of $\ch(\Phi_0) = \frac{1}{d^2}I_{d^2}$ we can take $I_{d^2} = \frac{1}{d^2}\sum_{k=0}^{d-1} \sum_{l=0}^{d-1} |k\ra\la k| \otimes | l \ra\la l |$. One can see that this is an example of why Theorem \ref{th:muf-1} doesn't hold for $t=0$. 

But there are many other decompositions in this case. For real weights $w_i \ge 0$ that sum to $1$ and unit vectors $\{|x_i\ra\}_{i=1}^{d^2}, \{|y_i\ra\}_{i=1}^{d^2} \subset \mbb{C}^d$ if we have 

\ea{\label{eq:t0-1}
\sum_{i=1}^{d^2} w_i |x_i\ra\la x_i| \otimes |y_i\ra\la y_i| 
= \frac{1}{d^2}I_{d^2}, 
}
then $w_i = \frac{1}{d^2}$ and $\{|x_i\ra|y_i\ra\}$ can be any orthonormal basis of $\mbb{C}^{d} \otimes \mbb{C}^{d}$. 
As an example of this we can take any two orthonormal bases $\{|a_k\ra\}_{k=1}^d$, $\{|b_l\ra\}_{l=1}^d$ of $\mbb{C}^d$ 
and set $|x_{(k,l)}\ra = |a_k\ra$, $|y_{(k,l)}\ra = |b_l\ra$, where we switched to indexing over $\{1,\dots,d\}^2$. 
It's an interesting question to find a concrete description of the set of bases of $\mbb{C}^{d} \otimes \mbb{C}^{d}$, where every basis vector is a simple tensor.

A similar question one can ask about Weyl-Heisenberg group covariant solutions for Eq. \ref{eq:t0-1}.

\newpage 

The Weyl-Heisenberg group is a group of unitary matrices generated by $W_\mathbf{a} = \tau^{a_1 a_2} S^{a_1}C^{a_2}$ 
for every $\mathbf{a} = (a_1,a_2) \in \mbb{Z}^2$, where 
$$
	\tau = \exp(2\pi i \cdot \frac{d+1}{2d}), ~~~ \omega = \tau^2 = \exp(2\pi i \cdot \frac{1}{d}),
$$
\ea{\label{eq:t0-2}
	S = \sum_{i=0}^{d-1} |i+1\ra\la i |,  ~~~ C = \sum_{i=0}^{d-1} \omega^i |i \ra\la i |,
}
(we assume that $|d\ra = |0\ra$). 
For any $\mathbf{a},\mathbf{b} \in \mbb{Z}^2$ they satisfy relations 

$$
  W_\mathbf{a}^\dag = W_\mathbf{a}^{-1} = W_{-\mathbf{a}},
$$
$$
W_\mathbf{a} W_\mathbf{b} = \tau^{\la \mathbf{a},\mathbf{b} \ra} W_\mathbf{a+b},
$$
\ea{\label{eq:t0-3}
	W_{\mathbf{a} + d \mathbf{b}} = \tau^{d\la \mathbf{a},\mathbf{b}\ra} W_\mathbf{a}, 
}
where the symplectic form $\la \mathbf{a},\mathbf{b} \ra$ is defined
by
\begin{equation}\label{eq:t0-4}
  \la \mathbf{a},\mathbf{b} \ra = a_2b_1 - b_2a_1 = - \la
  \mathbf{b},\mathbf{a} \ra.
\end{equation}

See \cite{Appleby} for more details.

We say that a frame of size $d^2$ is WH covariant if it consists of vectors $W_\mathbf{a}|x\ra$ for $\mathbf{a} \in \{ 0, \dots, d-1\}^2$, 
where $|x\ra$ is a \textit{fiducial} unit vector. 
The strong version of Zauner's conjecture states that there is always a WH covariant SIC in any dimension. 

We're concerned about solutions of Eq. \ref{eq:t0-1} with $|x_{\mathbf{a}}\ra = W_{\mathbf{a}} |x\ra$ and $|y_{\mathbf{a}}\ra = W_{\mathbf{a}} |y\ra$, where we switched to indexing over $\mathbf{a} \in \{0,\dots,d-1\}^2$. 
One such solution has already appeared in \cite{OY}, Proposition 4.  
If $F$ is the discrete Fourier transform matrix, that is, $F = \frac{1}{\sqrt{d}} \sum_{kl} w^{kl} |k \ra\la l|$, then we can set $|x\ra = |0\ra$ and $|y\ra = F |0\ra$. 
It's not hard to check that $\{W_{\mathbf{a}} |0\ra \otimes  W_{\mathbf{a}} F|0\ra \}$ is an orthonormal basis of $\mbb{C}^{d} \otimes \mbb{C}^{d}$. 
Moreover, the pair of frames $\{W_{\mathbf{a}} |0\ra\}$, $\{W_{\mathbf{a}} F|0\ra \}$ satisfy the defining relations for a pair of MUF for $t=0$, with the exception of informational completeness. 

Although, it turns out that this "good looking" solution at $t=0$ isn't a good fit for constructing a continuation. 
Recall that in dimensions 2 and 3 the families of solutions found in \cite{PPPR} are WH covariant and continuous over $t$ on the interval $[0, \frac{1}{d+1}]$. Moreover, one can verify that those families of solutions are differentiable at $t=0$. 
We prove that this can't be the case if a continuous family of covariant solutions has $\{W_{\mathbf{a}} |0\ra\}$, $\{W_{\mathbf{a}} F|0\ra \}$ at $0$. 

\begin{proposition}\label{prop:t0-4}
Let $d$ be any dimension. Suppose that there exists a continuous WH covariant family of pairs of frames $\{W_{\mathbf{a}} |x(t)\ra\}$, $\{W_{\mathbf{a}} |y(t)\ra\}$ in $\mbb{C}^d$ for $t \in [0, \epsilon]$, $\epsilon >0$, that satisfy equation
\ea{\label{eq:t0-5}
\frac{1}{d^2} \sum_{\mathbf{a} \in \{0,\dots,d-1\}^2} W_\mathbf{a}^{\otimes 2} \big(|x(t)\ra\la x(t)| \otimes |y(t)\ra\la y(t)|\big) W_\mathbf{a}^{\otimes2 \dag} 
= \frac{t}{d}U_{SW} +\frac{1-t}{d^2}I_{d^2}. 
} 
Also, assume that $|x(0)\ra = |0 \ra$, $|y(0)\ra = F|0\ra$. 
Then such a family can't be differentiable at $t=0$. 

\end{proposition}

To prove this we need the following lemma 

\begin{lemma}\label{lem:t0-5}
For any matrix $X \in \mbb{C}^{d^2 \times d^2} $ it holds that 
\ea{\label{eq:t0-6}
\sum_{\mathbf{a} \in \{0,\dots,d-1\}^2}  W_\mathbf{a}^{\otimes 2} X W_\mathbf{a}^{\otimes2 \dag}  
= \sum_{\mathbf{a} \in \{0,\dots,d-1\}^2} \tr(X W_\mathbf{a}^\dag \otimes W_\mathbf{a}) W_\mathbf{a} \otimes W_\mathbf{a}^\dag. 
} 

\end{lemma}
\begin{proof}
One can see that both sides are linear over $X$. 
Thus, it's enough to prove it for $X = W_\mathbf{b} \otimes W_\mathbf{c}$ for any $\mathbf{b}, \mathbf{c} \in \{0,\dots,d-1\}^2$. 

For the left side we have 
$$
\sum_{\mathbf{a} \in \{0,\dots,d-1\}^2}  W_\mathbf{a}^{\otimes 2} (W_\mathbf{b} \otimes W_\mathbf{c}) W_\mathbf{a}^{\otimes2 \dag}  
= \sum_{\mathbf{a} \in \{0,\dots,d-1\}^2} \omega^{\la a , b \ra} W_\mathbf{b} \otimes \omega^{\la a , c \ra} W_\mathbf{c} 
=
$$
\ea{\label{eq:t0-7}
= \Big(\sum_{\mathbf{a} \in \{0,\dots,d-1\}^2} \omega^{\la \mathbf{a} , \mathbf{b+c} \ra}\Big)  W_\mathbf{b} \otimes W_\mathbf{c} 
=  \begin{cases}
    0, \text{ if } \mathbf{b+c} \neq (0,0),
    \\ 
   d^2 W_\mathbf{b} \otimes W_\mathbf{-b}, \text{ if } \mathbf{b+c} = (0,0),
   \end{cases}
}
where the result of summation of $\omega$ powers we deduced from the fact that $\sum_{\mathbf{a}} W_\mathbf{a}W_\mathbf{z}W_\mathbf{a}^\dag = d \tr(W_\mathbf{z})$, which is $0$ unless $\mathbf{z}=(0,0)$. 

For the right side we have 
$$
\sum_{\mathbf{a} \in \{0,\dots,d-1\}^2} \tr(W_\mathbf{b} \otimes W_\mathbf{c}  \cdot W_\mathbf{a}^\dag \otimes W_\mathbf{a}) W_\mathbf{a} \otimes W_\mathbf{a}^\dag 
= \sum_{\mathbf{a} \in \{0,\dots,d-1\}^2} \tr(W_\mathbf{b}W_\mathbf{-a}) \tr(W_\mathbf{c}W_\mathbf{a}) W_\mathbf{a} \otimes W_\mathbf{a}^\dag
=
$$
\ea{\label{eq:t0-8}
= \begin{cases}
    d^2 W_\mathbf{a} \otimes W_\mathbf{-a} , \text{ if } \mathbf{a}=\mathbf{b} \text{ and } \mathbf{a}=\mathbf{-c},
    \\ 
   0, \text{ otherwise},
   \end{cases}
}
which is the same as for the left side. 
\end{proof}
Note that the expression Eq. \ref{eq:t0-6} is the scaled projection of the matrix $X$ onto the subspace spanned by the matrices $W_\mathbf{a} \otimes W_\mathbf{a}^\dag$ (which are orthogonal to each other). 

Let's prove Prop. \ref{prop:t0-4}.

\begin{proof}

Denote $\pi_0 = |x(0)\ra\la x(0)| = |0\ra\la0|$,  $\rho_0 = |y(0)\ra\la y(0)| = F|0\ra\la0|F^\dag$ and
assume that $\pi'_0 = \frac{\mrm{d}}{\mrm{d} t} |x(t)\ra\la x(t)| \mid_{t=0}$ and $\rho'_0 = \frac{\mrm{d}}{\mrm{d} t} |y(t)\ra\la y(t)| \mid_{t=0}$ exist (here we're considering the right derivative). 
Differentiation of Eq. \ref{eq:t0-5} gives us 

\ea{\label{eq:t0-9}
\frac{1}{d^2} \sum_{\mathbf{a} \in \{0,\dots,d-1\}^2} W_\mathbf{a}^{\otimes 2} \big( \pi_0 \otimes \rho'_0 + \pi'_0 \otimes \rho_0 \big) W_\mathbf{a}^{\otimes2 \dag} 
= \frac{1}{d}U_{SW} - \frac{1}{d^2}I_{d^2}.  
} 

From \cite{OY}, Lemma 1, one can deduce that $U_{SW} = \frac{1}{d} \sum_\mathbf{a} W_\mathbf{a} \otimes W_\mathbf{a}^\dag$, 
hence 
\ea{\label{eq:t0-10}
\frac{1}{d}U_{SW} - \frac{1}{d^2}I_{d^2} = \frac{1}{d^2} \sum_{\mathbf{a} \neq (0,0)} W_\mathbf{a} \otimes W_\mathbf{a}^\dag.
}

By using Lemma \ref{lem:t0-5} we conclude that $\tr(\pi_0 \otimes \rho'_0 + \pi'_0 \otimes \rho_0) = 0$ and for any $\mathbf{a} \neq (0,0)$:
\ea{\label{eq:t0-11}
	\tr\Big((\pi_0 \otimes \rho'_0 + \pi'_0 \otimes \rho_0) W_\mathbf{a}^\dag \otimes W_\mathbf{a}\Big) = 1.
}

But if $a_1 \neq 0$ and $a_2 \neq 0$ then for $\mathbf{a} = (a_1,a_2)$ we have $\tr(\pi_0 W_\mathbf{a}) = \la 0| W_\mathbf{a} |0\ra = 0$ and $\tr(\rho_0 W_\mathbf{a}) = \la 0| F^\dag W_\mathbf{a} F | 0 \ra = \la 0| W_{(a_2,-a_1)} | 0 \ra = 0$, hence $\tr\big((\pi_0 \otimes \rho'_0 + \pi'_0 \otimes \rho_0) W_\mathbf{a}^\dag \otimes W_\mathbf{a}\big) = 0$ – a contradiction. 

\end{proof}

You may ask what WH covariant solutions at $0$ evade this proof. 
Let's take $|x(0)\ra = |0\ra$ and $|y(0)\ra = \sum_i \gamma_i |i\ra$, where $|\gamma_i| = \frac{1}{\sqrt{d}}$, 
also we require that $\la y(0) | W_{(a_1,a_2)} |y(0)\ra \neq 0$ for any $a_1 \neq 0$ and any $a_2$.   
It's not hard to check that such $\gamma_i$ exist, 
the pair of frames $\{W_\mathbf{a}|x(0)\ra\}$, $\{W_\mathbf{a}|y(0)\ra\}$ satisfy Eq.\ref{eq:t0-1} and MUF relations for $t=0$, 
except the first frame is not informationally complete. 
The families of solutions found in \cite{PPPR} at $t=0$ have exactly this type. 

\section{Conclusions}

In this paper we showed that there is a deeper connection between Zauner's conjecture and the entanglement breaking rank of the quantum depolarizing channels. We hope that this could shed some light on the original conjecture, and on the research of entanglement breaking maps and their rank in general.

\section{Acknowledments}
The author is grateful to Vasyl Ostrovskyi for useful discussions, suggestions and support.

\end{document}